\documentclass[11pt, letterpaper]{article}

\usepackage{fullpage}
\usepackage{amsthm}
\usepackage{amsmath,amssymb,amsfonts,nicefrac}
\usepackage{xspace}
\usepackage{color}
\usepackage{url}
\usepackage{bm}
\usepackage{bbm}
\usepackage{times}
\newtheorem{theorem}{Theorem}[section]

\newtheorem{lemma}[theorem]{Lemma}

\newtheorem{claim}[theorem]{Claim}

\newcommand{\junk}[1]{}
\newcommand{\ignore}[1]{}
\newcommand \E {{\rm E}}

\newcommand{\eps}{\varepsilon}

\newcommand{\initOneLiners}{%
    \setlength{\itemsep}{0pt}
    \setlength{\parsep }{0pt}
    \setlength{\topsep }{0pt}
}

\def\showauthornotes{1}
\ifnum\showauthornotes=1
\newcommand{\Authornote}[2]{{\sf\small\color{red}{[#1: #2]}}}
\else
\newcommand{\Authornote}[2]{}
\fi

\begin{document}
\title{A PTAS for the Classical Ising Spin Glass Problem  \\ on the
Chimera Graph Structure}
\author{Rishi Saket\thanks{IBM T. J. Watson Research Center. Email: {\tt
rsaket@us.ibm.com}}}
\maketitle
\begin{abstract}
We present a polynomial time approximation scheme (PTAS) for
the minimum value of the classical Ising Hamiltonian with linear terms
on the \emph{Chimera} graph structure as defined in the recent work of 
McGeoch and
Wang~\cite{MW13}. 
The result follows from a direct application of the techniques used by
Bansal, Bravyi and Terhal~\cite{BBT} who gave a PTAS for the same
problem on planar and, in particular, grid graphs. We also show that
on Chimera graphs, the trivial lower bound is within a constant factor
of the optimum.
\end{abstract}

\section{Introduction}
The classical Ising spin glass problem is defined as follows. 
Given a graph $G(V, E)$ along with real numbers $d_u$ for all vertices
$u$
and  $c_{uv}$ for all edges $(u,v)$,  
the classical Ising spin glass problem is to
the minimize the following Hamiltonian,
\begin{equation}
H(S) := \displaystyle\sum_{(u,v)\in E}c_{uv}S_uS_v +
\displaystyle\sum_{u\in V}d_uS_u, \label{eqn-hamil}
\end{equation}
over all $\{-1,1\}$ assignments to the vertices given by 
$S = \{S_u \in \{-1,1\}\}_{u\in V}$. It is useful to note that 
$\E_S[H(S)] = 0$, and thus the optimum of $H(S)$ is non-positive.

The Ising spin glass problem is used to model the interactions in 
physical spin systems and its minimum 
value is a measure of the ground state-energy 
of the system. As such this computational problem has received
significant attention from both, the algorithmic and complexity
perspectives. For a detailed
discussion on this problem we refer the reader to the related work of
Bansal, Bravyi and Terhal~\cite{BBT} who
gave a PTAS for this
problem on planar graphs.
In this work we focus on the
approximability of this problem on the \emph{Chimera} graph structure
which we formally define below. 

\medskip
\noindent
{\bf The Chimera Graph Structure.}
We shall work with the Chimera graph structure as defined in
\cite{MW13}, with a different notation for convenience. For a
positive integer $r$, the Chimera graph $G_r$ is constructed as follows.

\medskip
\noindent
\emph{Vertices $V$.} The set of vertices $V$ consists of the
integer tuples $\{(i,j,k,l)\in \mathbb{Z}^4\ \mid \ 1\leq i,j\leq r, 1\leq k
\leq 4, l=0,1\}$. The indices $i$ and $j$ specify the location
in a
$r\times r$ grid, $k$ indicates
one of four possible vertices, and
$l$ specifies the \emph{layer} in which the vertex belongs. The number of
vertices is $8r^2$.

\medskip
\noindent
\emph{Edges $E$.} The edge set is a disjoint union of $E_0, E_1$ and
$E_{01}$, where $E_0$ is the set of edges in layer $0$, $E_1$ is the
set of edges in layer $1$, and $E_{01}$ is the set of edges across the
the two layers. These sets of edges are defined as follows:
\begin{itemize}
\item[$E_{0}$:] For any $1\leq i\leq r-1$, $1\leq j\leq r$,  
and $1\leq k\leq 4$, 
$E_0$ contains an edge between
$(i,j,k,0)$ and $(i+1,j,k,0)$. 
\item[$E_1$:]  For any $1\leq i\leq r$, $1\leq j\leq r-1$,  
and $1\leq k\leq 4$, 
$E_1$ contains an edge between
$(i,j,k,1)$ and $(i,j+1,k,1)$.
\item[$E_{01}$:] For any $1\leq i\leq r$, $1\leq j\leq r$, 
$1\leq k_0\leq 4$ and $1\leq k_1\leq 4$, $E_{01}$ contains an edge
between $(i,j,k_0,0)$ and $(i,j,k_1,1)$.
\end{itemize}
Note that $E_0$ is a disjoint
collection of $4r$   $(r-1)$-length paths, one each for a fixed
pair of values for $j$ and $k$. Similarly, $E_1$ is a disjoint
collection of $4r$  $(r-1)$-length paths, one each for a fixed
pair of values for $i$ and $k$. Also,
$E_{01}$ is a disjoint
 collection of $r^2$ $K_{4,4}$ graphs, one for each value of
$(i,j)$ with the bipartition given by the sets
$\cup_{k=1}^4\{(i,j,k,0)\}$ and $\cup_{k=1}^4\{(i,j,k,1)\}$.

Note that the Chimera graph structure is non-planar as it contains the
$K_{4,4}$ graph, as well as a $K_{r,r}$ minor. Thus, 
the results of Bansal et al.~\cite{BBT} are not
directly applicable.
However, utilizing the symmetries in the above construction we are
able to adapt the techniques of \cite{BBT} to prove the
following algorithmic results. 

\medskip
\noindent
{\bf Our Results.} This work shows the existence of a polynomial time
approximation scheme (PTAS) for the classical Ising spin glass problem
on the Chimera graph structure. Formally we prove the following
theorem.
\begin{theorem}\label{thm-main1} 
Given a Chimera graph structure $G_r$ on $n=8r^2$
vertices, the Hamiltonian $H(S)$ can be approximated to $(1-\eps)$ of
its minimum value in time 
$O\left(n\cdot 2^{\frac{32}{\eps}}\right)$.
\end{theorem}
The above result is obtained by noting that the graph $G_r$ can be
disconnected into constant width \emph{strips} by removing 
a small fraction of edges from $E_0$
(or $E_1$). This allows the application of the partitioning 
technique used by Bansal et
al.~\cite{BBT} for their PTAS on grid graphs. 
The analysis requires a
straightforward lower bound on the magnitude of the optimum value of $H(S)$ in
terms of sum of the absolute values of the bilinear coefficients in
$H(S)$ corresponding to edges in $E_0$.

A somewhat more involved analysis yields the following result which shows
that the trivial lower bound is within a constant factor of the optimum.
\begin{theorem} \label{thm-main2} 
Let $H^*$ be the optimum value of $H(S)$ on the
Chimera graph structure. Then, $H^* \leq 
-(C/(3C+4))\left[\sum_{(u,v)\in E}|c_{uv}| + \sum_{u\in V} |d_u|\right]$, 
for some constant $C > \frac{\ln(1
+ \sqrt{2})}{\pi}$. In particular, the approximation factor is
$(3C+4)/C < 17.26$.  
\end{theorem}
The above result is obtained by extending the straightforward lower bound
used to prove Theorem \ref{thm-main1} with a complementary 
bound obtained via the
Grothendieck constant.

\section{PTAS on Chimera Graphs}
Let $G_r = G(V,E)$ be the Chimera graph on $n=8r^2$ vertices and $H(S)$ be
the Hamiltonian given in Equation \eqref{eqn-hamil} for some real
values $\{c_{uv}\}_{(u,v)\in E}$ and $\{d_u\}_{u\in V}$. Recall that
$E$ is the disjoint union, $E = E_0\cup E_1\cup E_{01}$. 
For convenience we split $H(S)$
as, 
\begin{equation}
H(S) := M_0(S) + M_1(S) + M_{01}(S) + D(S),
\end{equation}
where,
\begin{align}
M_l(S) := & \sum_{(u,v)\in E_l}c_{uv}S_uS_v,  \ \ \ \ l=0,1, \nonumber
\\
M_{01}(S) := & \sum_{(u,v)\in E_{01}}c_{uv}S_uS_v, \nonumber \\
D(S) := & \sum_{u\in V}d_uS_u. \nonumber
\end{align}
We also define the following quantities.
\begin{align}
A_l := & \sum_{(u,v)\in
E_l}|c_{uv}|\ \ \  \textnormal{for}\ l = 0,1. \nonumber \\
A_{01} := & \sum_{(u,v)\in
E_{01}}|c_{uv}|. \nonumber \\
B := & \sum_{u\in V}|d_u|.
\end{align}
The following
lemma follows from the structure of $G$.
\begin{lemma} \label{lem-A0A1}
Let $H^*$ be the minimum value of $H(S)$. Then, $H^*
\leq -(A_0 + A_1)$.
\end{lemma}
\begin{proof}
From the  construction of $G_r$ 
we have that $E_0$ is a disjoint collection of $4r$
$(r-1)$-length paths in layer $0$. Similarly,  $E_1$ is a disjoint 
collection of $4r$ $(r-1)$-length paths in layer $1$. Thus, $E_0\cup
E_1$ is a disjoint collection of $8r$ $(r-1)$-length paths. Thus,
there exists an assignment $S'$ such that $M_0(S') = -A_0$ and
$M_1(S') = -A_1$. We can ensure that $M_{01}(S') \leq 0$, otherwise
the values assigned to all the vertices in layer $0$ -- i.e. all
vertices of the form $(i,j,k,0)$ -- can be flipped which changes the sign
of $M_{01}(S')$ while preserving $M_0(S')$ and $M_1(S')$. Thus, we can
ensure that $M_0(S) + M_1(S) + M_{01}(S) \leq -(A_0 + A_1)$. Now, if
$D(S')$ is positive, then $S'$ can be flipped for all vertices to
ensure that $H(S') \leq  -(A_0 + A_1)$.
\end{proof}

Let us define a \emph{strip} graph with $m$ levels and width $b$ as
any graph which has $m$ levels of $b$ vertices each such that all
edges between levels are between adjacent levels. The levels may have
edges within them. To be precise, a vertex in level $j$ may have
edges only to other vertices in levels $j$, $(j-1)$ or $(j+1)$. Bansal
et al.~\cite{BBT} showed that the problem of minimizing the
Hamiltonian on $m\times b$ strip graphs 
can be solved using dynamic programming in time $O(m4^{b})$. The
dynamic program computes for level $j$ and each of the $2^{b}$ assignments
to the vertices in that level, the value of the best solution with
that assignment to level $j$ for the Hamiltonian on the
subgraph induced by levels $1,\dots, j$. Going from level $j$ to $j+1$
requires $O(4^{b})$ operations, a constant number for each 
pair of assignments to
levels $j$ and $j+1$. Thus, the total time taken is $O(m4^{b})$. Our
goal in designing the PTAS is to describe a way to decompose the graph
$G$ into strip graphs with constant width while not losing much in the
objective value. This is similar to the decomposition in \cite{BBT}
for grid graphs with a slightly different analysis.

Recalling the structure of $G_r$, for any $i = 1,\dots, r$ let, 
$$E_{i0} := \{((i,j,k,0), (i',j,k,0))\ \mid\ i'=\{i-1,i+1\}\cap\{1,\dots,
r\},\ 1\leq j\leq r,\ 1\leq k\leq 4\}.$$
In other words, $E_{i0}$ consists of all edges within layer $0$ which
are incident on the vertices $(i,j,k,0)$. 
Let $T$ be a large positive integer which we shall set later. For any
$k = 0, 1, \dots, T-1$, let,
\begin{align}
E^k_0 := \ & \displaystyle\bigcup_{\substack{i\ \equiv\ k\mod T, \\
1\leq i\leq r}}E_{i0}, \nonumber \\
A^k_0 := \ &\displaystyle\sum_{(u,v)\in E^k_0}|c_{uv}|, \nonumber \\
H_k(S) := \ &\displaystyle\sum_{(u,v)\in E^k_0}c_{uv}S_uS_v, \nonumber
\\
H_{sub,k}(S) :=\ & H(S) - H_k(S). \nonumber
\end{align}

Since every edge in $E_0$ is present in exactly $2$ of the subsets
$E^k_0$, we obtain by averaging that there is a $k^*$ such that
$A^{k^*}_0 \leq (2/T)A_0$. Further, it can be seen that $H_{sub,k}(S)$
is the Hamiltonian on a disjoint collection of at most $\lceil
r/T\rceil$ strip graphs of dimensions $r\times 8T$ 
 and at most $\lceil r/T\rceil$ strip graphs of dimensions 
$r\times 8$. To complete the
analysis we first assume because of symmetry that $A_1 \geq A_0$. Thus,
by Lemma \ref{lem-A0A1},
the minimum value of $H(S)$ is at most $-2A_0$. Let $S'$ be the
assignment that minimizes the value of $H_{sub,k^*}$, and
$S^{\textnormal{opt}}$ the assignment that minimizes the value of
$H(S)$. We have,
\begin{align}
H_{sub,k^*}(S') \leq  & H_{sub,k^*}(S^{\textnormal{opt}}) \nonumber \\
		\leq & H(S^{\textnormal{opt}}) + A^{k^*}_0 \nonumber
\\
		\leq & H(S^{\textnormal{opt}}) + (2/T)A_0 \nonumber \\
		\leq & (1 - 1/T) H(S^{\textnormal{opt}}), \nonumber
\end{align}
and,
\begin{align}
H(S')           \leq  & H_{sub,k^*}(S') + A^{k^*}_0 \nonumber \\
		\leq & (1 - 1/T) H(S^{\textnormal{opt}}) + (2/T)A_0 \nonumber \\
		\leq & (1 - 2/T) H(S^{\textnormal{opt}}). \nonumber
\end{align}
Since $H(S') \geq H(S^{\textnormal{opt}})$, 
from the above we get that $H(S')$ is a $(1 -
2/T)$ approximation to $H(S^{\textnormal{opt}})$. We now set $T =
2/\eps$ to obtain a $(1-\eps)$ approximation to
$H(S^{\textnormal{opt}})$. To compute the value $H(S')$, we compute
the assignment $S'$ that minimizes $H_{sub,k}(S)$ for each $k = 0, 1,
\dots, T-1$, and take the one that gives the minimum value of $H(S')$.
For each $k$, this involves minimizing the Hamiltonian on at most $(2r/T +
2)$ strip graphs of $r$ levels and width at most $8T$. As discussed
above, the computation time for one such strip graph is
$O(r4^{8T})$ which adds up to $O(r^24^{8T}/T)$ time for all the 
strip graphs for one value of $k$. Thus, the total
computation time to obtain a $(1-\eps)$ approximation is
$O\left(r^24^{\frac{16}{\eps}}\right) = 
O\left(n2^{\frac{32}{\eps}}\right)$. 

\section{A constant factor bound}
In this section we prove Theorem \ref{thm-main2}. The key ingredient
is the following lemma which bounds the bilinear quadratic form on
$K_{4,4}$ in terms of the sum of the absolute values of the
coefficients. 
\begin{lemma} \label{lem-k44}
Let $G(U, V, E)$ be a $K_{4,4}$ graph with $U = \{u_1,
\dots, u_4\}$ and $V = \{v_1, \dots, v_4\}$ giving the bipartition and
$E = U\times V$ the edge set. There is a universal constant $C > 0$
such that for any real numbers $\{c_{ij}\mid 1\leq i,j\leq 4\}$,
\begin{equation}
\min_S\sum_{1\leq i,j\leq 4}c_{ij}S_{u_i}S_{v_j} \leq - C\sum_{1\leq i,j\leq
4}|c_{ij}|,
\end{equation}
where $S = \{S_{u}\mid u\in U\}\cup\{S_v\mid v\in V\}$ 
is a $\{-1,1\}$ assignment to the vertices of $G$. 
In particular, the above holds for some $C > \frac{\ln(1+\sqrt{2})}{\pi}$.
\end{lemma}
\begin{proof}
Let us first assign unit vectors $x_i$ for vertices $u_i$ and $y_j$
for vertices $v_j$ ($1\leq i,j\leq 4$). The seminal work of
Grothendieck~\cite{Groth} implies the following:
\begin{equation}
\max_{\substack{x_i, y_j \in \textnormal{\bf S}^{7} \\ 1\leq i,j\leq
4}} \sum_{1\leq i,j\leq 4}c_{ij}\langle x_i, y_j\rangle \leq K
\max_{S} \sum_{1\leq i,j\leq 4}c_{ij}S_{u_i}S_{v_j}. \label{eqn-groth}
\end{equation}
Here $K$ is a universal constant for which the above inequality holds 
for any $K_{t,t}$, wherein our case $t=4$. Determining the exact value of $K$
has been a major open question. The work of Krivine~\cite{Kriv} showed
that $K \leq \frac{\pi}{2\ln(1+\sqrt{2})}$ and in more recent work
Braverman, Makarychev, Makarychev and Naor~\cite{BMMN} showed that in
fact $K <  \frac{\pi}{2\ln(1+\sqrt{2})}$. 
The following claim helps to leverage the above inequality.
\begin{claim}\label{claim-1}
There exists a set of unit vectors vectors $\{x_i \mid 1\leq i\leq
4\}\cup\{y_j \mid 1\leq j\leq 4\}$ such that,
\begin{equation}
\displaystyle\sum_{1\leq i,j\leq 4}c_{ij}\langle x_i, y_j\rangle \geq
\left(\frac{1}{2}\right) 
\displaystyle\sum_{1\leq i,j\leq 4}|c_{ij}|.
\end{equation}
\end{claim}
\begin{proof}
We first set the vectors $\{x_i \mid 1\leq i \leq 4\}$ to be a set of
$4$ orthonormal vectors. For $j = 1, 2, 3, 4$ we let the unit
vector $y_j = (1/2)\sum_{i=1}^4 \textnormal{sgn}(c_{ij})x_i$. It is
easy to check that this setting of the vectors satisfies the
inequality in the claim.
\end{proof}
Using the above claim in conjunction with Equation \eqref{eqn-groth}
along with the bound on the value of $K$
we obtain,
\begin{equation}
\max_{S} \sum_{1\leq i,j\leq 4}c_{ij}S_{u_i}S_{v_j} \geq 
 C\sum_{1\leq i,j\leq
4}|c_{ij}|
\end{equation}
Reversing the signs of the assignments to the vertices in $U$, we
complete the proof of the lemma.
\end{proof}
The following lemma provides a bound on the minimum value of $H(S)$. 
\begin{lemma}
Let $H^*$ be the minimum value of $H(S)$. Then $H^* \leq A_0 + A_1 -
CA_{01}$.
\end{lemma}
\begin{proof}
Since the set of edges $E_{01}$ is a disjoint collection of $r^2$
$K_{4,4}$ graphs, using Lemma \ref{lem-k44} 
one can set the assignment $S$ such that $M_{01}(S) \leq -CA_{01}$.
The maximum values of $M_0(S)$ and $M_1(S)$ are $A_0$ and $A_1$
respectively. Also, by  flipping the sign of $S$ if necessary, 
one can simultaneously
ensure that $D(S)$ is non-positive.  
\end{proof}
Combining the above lemma with Lemma \ref{lem-A0A1} we obtain,
\begin{align}
(C + 2)H^* & \leq -(C+1)(A_0 + A_1) + A_0 + A_1 -
CA_{01}, \nonumber \\
\Rightarrow (C+2)H^* & \leq -C(A_0 + A_1 +
A_{01}),
\nonumber \\
\Rightarrow \left(\frac{C+2}{C}\right)H^* & \leq -(A_0 + A_1 +
A_{01}). \label{eqn-edgebd}
\end{align}
Moreover, an appropriate assignment of $S$ ensures that $D(S) = -B$.
Thus, we also have the following bound,
\begin{equation}
H^* \leq A_0 + A_1 + A_{01} - B. \nonumber
\end{equation}
Combining the above with Equation \eqref{eqn-edgebd} yields,
\begin{align}
\left(\frac{2(C+2)}{C} + 1 \right)H^* \leq & 
-2(A_0 + A_1 +
A_{01}) + (A_0 + A_1 + A_{01} - B), \nonumber \\
\Rightarrow H^* \leq & -\left(\frac{C}{3C + 4}\right)(A_0 + A_1 + 
A_{01} + B) \nonumber
\end{align}  
which completes the proof of Theorem \ref{thm-main2}.
\bibliographystyle{alpha}
\bibliography{Refs}

\begin{thebibliography}{BMMN11}

\bibitem[BBT09]{BBT}
N.~Bansal, S.~Bravyi, and B.~M. Terhal.
\newblock Classical approximation schemes for the ground-state energy of
  quantum and classical ising spin hamiltonians on planar graphs.
\newblock {\em Quantum Information {\&} Computation}, 9(7):701--720, 2009.

\bibitem[BMMN11]{BMMN}
M.~Braverman, K.~Makarychev, Y.~Makarychev, and A.~Naor.
\newblock The {Grothendieck} constant is strictly smaller than {Krivine's}
  bound.
\newblock In {\em Proceedings of the Annual Symposium on Foundations of
  Computer Science}, pages 453--462, 2011.

\bibitem[Gro53]{Groth}
A.~Grothendieck.
\newblock R\'esum\'e de la th\'eorie m\'etrique des produits tensoriels
  topologiques.
\newblock {\em Bol. Soc. Mat. S\~ao Paulo}, 8:1--79, 1953.

\bibitem[Kri77]{Kriv}
J.~L. Krivine.
\newblock Sur la constante de grothendieck.
\newblock {\em C. R. Acad. Sci. Paris S\'er. A-B}, 284(8):A445--A446, 1977.

\bibitem[MW13]{MW13}
C.~C. McGeoch and C.~Wang.
\newblock Experimental evaluation of an adiabiatic quantum system for
  combinatorial optimization.
\newblock In {\em Proc. ACM Computing Frontiers}, 2013.

\end{thebibliography}

\end{document}